\documentclass[11pt]{article}       

\usepackage{graphicx}
\usepackage{amsfonts}
\usepackage[english]{babel}
\usepackage{amsmath}
\usepackage{amsthm}
\usepackage{amssymb}
\usepackage{multirow}
\usepackage{algorithm}
\usepackage[noend]{algpseudocode}
\newtheorem{thm}{Theorem}

\newtheorem{cor}[thm]{Corollary}
\newtheorem{rem}[thm]{Remark}

\begin{document}

\title{Nonexistence of a few binary orthogonal arrays
\thanks{The research of the first author was supported, in part, by a Bulgarian NSF contract I01/0003.
        The research of the second and third authors was supported, in part,
        by the Science Foundation of Sofia University under contract 144/2015
and its continuation ''Discrete, Algebraic and Combinatorial Structures'' in 2016.}
}

\author{Peter Boyvalenkov \\
Institute of Mathematics and Informatics, \\
Bulgarian Academy of Sciences, \\ 
8 G.Bonchev Street, 1113, Sofia, BULGARIA \\ [3pt]
and Faculty of Mathematics and Natural Sciences, \\ 
South-Western University, Blagoevgrad, Bulgaria. \\
email: {peter@math.bas.bg} \\
Tanya Marinova, Maya Stoyanova \\
Faculty of Mathematics and Informatics, Sofia University, \\
5 James Bourchier Blvd., 1164 Sofia, BULGARIA \\
email: {tanya.marinova@fmi.uni-sofia.bg} \\ 
email: {stoyanova@fmi.uni-sofia.bg} \\
}

\date{} 

\maketitle

\begin{abstract}
We develop and apply combinatorial algorithms for
investigation of the feasible distance distributions of binary orthogonal arrays
with respect to a point of the ambient binary Hamming space utilizing
constraints imposed from the relations between the distance distributions of
connected arrays. This turns out to be strong enough and we prove the nonexistence of binary orthogonal arrays of parameters
(length, cardinality, strength)$\ =(9,6.2^4=96,4)$, $(10,6.2^5,5)$, $(10,7.2^4=112,4)$,
$(11,7.2^5,5)$, $(11,7.2^4,4)$ and $(12,7.2^5,5)$, resolving the
first cases where the existence was undecided so far. For the existing
arrays our approach allows substantial reduction of the number of feasible distance distributions
which could be helpful for classification results (uniqueness, for example).


\textbf{Subclass.}{MSC 05B15 \and 94B25 \and 94B65}
\end{abstract}

\section{Introduction}
\label{intro}

Orthogonal arrays have many connections to other combinatorial designs and have
applications in coding theory, the statistical design of experiments,
cryptography, various types of software testing and quality control.
We refer to the book \cite{HSS} as excellent exposition of the theory and
practice of orthogonal arrays. In fact, there are enormous material about orthogonal arrays in internet.

An orthogonal array (OA) of strength $\tau$ and index $\lambda$ in $H(n,2)$ (or binary orthogonal array, BOA),
consists of the rows of an $M\times n$ matrix $C$ with the property that every
$M \times \tau$ submatrix of $C$ contains all ordered $\tau$-tuples
of $H(\tau,2)$, each one exactly $\lambda=M/2^{\tau}$ times as rows.

Let $C \subset H(n, 2)$ be an $(n,M,\tau)$ BOA.
The distance distribution of $C$ with respect to $c \in H(n, 2)$ if the $(n+1)$-tuple
\[ w = w(c) = (w_0(c), w_1(c), \ldots, w_n(c)), \]
where $w_i(c) = |\{x \in C| d(x, c) = i\}|$, $i = 0, \ldots, n$.
All feasible distance distributions of BOA of parameters $(n,M,\tau)$ can be computed effectively for relatively small
$n$ and $\tau$ as shown in \cite{BK1}. Indeed, every distance distribution of $C$
satisfies the system
\begin{equation}
\label{eq in}
\sum_{i=0}^{n} w_{i}(c)\left(1-\frac {2i}{n}\right)^k=b_k |C|, \ \ \ k=0,1,\ldots,\tau,
\end{equation}
where $b_k=\frac{1}{2^n} \sum_{d=0}^n {n \choose d}\big( 1-\frac{2d}{n}\big)^k $
and, in particular, $b_k=0$ for $k$ odd.

The number $b_k$ is in fact the first coefficient in the
expansion of the polynomial $t^k$ in terms of (binary) Krawtchouk polynomials. The Krawtchouk
polynomials are zonal spherical functions for $H(n,2)$ (see \cite{DL,Lev95,Lev}) and can be the defined by
the three-term recurrence relation
\[ (n-k)Q_{k+1}^{(n)}(t)=ntQ_k^{(n)}(t)-kQ_{k-1}^{(n)}(t)\ \mbox{ for } 1 \leq k \leq n-1, \]
with initial conditions $Q_0^{(n)}(t)=1$ and $Q_1^{(n)}(t)=t$.

Let $n$, $M$ and $\tau \leq n$ be fixed. We denote by $P(n, M, \tau)$ the set of all possible
distance distributions of a $(n,M,\tau)$ BOA with respect to internal point $c$
(in the beginning -- all admissible solutions of the system \eqref{eq in} with $w_0(c) \geq 1$)
and by $Q(n, M, \tau)$ the set of all possible distance distributions of a $(n,M,\tau)$ BOA
with respect to external point (in the beginning -- all admissible solutions of the system \eqref{eq in} with $w_0(c)=0$).
Denote also $W(n,M,\tau) = P(n,M,\tau) \cup Q(n,M,\tau)$.

In this paper we describe an algorithm which works on the sets $P(n,M,\tau)$, $Q(n,M,\tau)$ and $W(n,M,\tau)$
utilizing connections between related BOAs. During the implementation of our algorithm these sets
are changed\footnote{However, we prefer to keep the initial notation
in order to avoid tedious notation.} by ruling out some distance distributions.

In Section \ref{sec:1} we prove several assertions which connect the distance
distributions of arrays under consideration and their relatives.
This imposes significant constraints on the targeted BOAs and therefore allows us to
collect rules for removing distance distributions from the sets
$P(n,M,\tau)$, $Q(n,M,\tau)$ and $W(n,M,\tau)$.
The logic of our algorithm is described in Section \ref{sec:2}. The new nonexistence results are described in
Section \ref{sec:3}.

Algorithms for dealing with distance distributions were proposed earlier in \cite{BK1} and \cite{BKMS}
but in these papers the set $P(n,M,\tau)$ was only examined. Moreover, two seemingly crucial observations
(Theorem \ref{t-complementary} together with Corollary \ref{complementary} and
Theorem \ref{a6} together with Corollary \ref{a6}) are new. Also, all complete versions
(for the set $W(n,M,\tau)$) of the remaining assertions from the next section are new.

\section{Relations between distance distributions of $(n,M,\tau)$ BOA and its derived BOAs}
\label{sec:1}

We start with a simple observation.

\begin{thm}
\label{t-complementary}
If the distance distribution $w=(w_0,w_1,\ldots,w_n)$ belongs to the set $W(n,M,\tau)$,
then the distance distribution $\overline{w}=(w_n,w_{n-1},\ldots,w_0)$ also belongs to $W(n,M,\tau)$.
\end{thm}

\begin{proof}
Let $C \subset H(n, 2)$ be a BOA of parameters $(n, M, \tau)$ and $\overline{C}$ is the array which
is obtained from $C$ by the permutation $(0 \to 1, 1 \to 0)$ in the whole $C$.
Since the distances inside $C$ are preserved by this transformation, $\overline{C}$ is again
$(n,M,\tau)$ BOA. On the other hand, distance $i$ from external for $C$ point to a point of $C$
correspond to distance $n-i$ to the transformed point of $\overline{C}$. This means that if
$w=(w_0,w_1,\ldots,w_n)$ is the distance distribution of $C$ with respect to some point $c \in H(n,2)$
(internal or external for $C$), then the distance distribution of $\overline{C}$ with respect to the same point
(which can become either internal or external for $\overline{C}$, depending on whether $w_n>0$ or $w_n=0$)
is $\overline{w}=(w_n,w_{n-1},\ldots,w_0)$. \hfill $\Box$
\end{proof}

\begin{cor}
\label{complementary}
The distance distribution $w=(w_0,w_1,\ldots,w_n) \in W(n, M, \tau)$ is ruled out if
$\overline{w}=(w_n,w_{n-1},\ldots,w_0) \not\in W(n, M, \tau)$.
\end{cor}

Corollary \ref{complementary} is important in all stages of our algorithm since it
requires the non-symmetric distance distributions to be paired off and infeasibility
of one element of the pair immediately implies the infeasibility for the other.

We proceed with analyzing relations
between the BOA $C$ and BOAs $C^\prime$ of parameters $(n-1,M,\tau)$ which are obtained
from $C$ by deletion of one of its columns. Of course, the set $W(n-1,M,\tau)$ of possible
distance distributions of $C^{\prime}$ is sieved by Corollary \ref{complementary} as well.

It is convenient to fix the removing of the first column of $C$. Let the
distance distribution of $C$ with respect to $c = \mathbf{0}=(0,0,\ldots,0) \in H(n,2)$ be
$w = (w_0,w_1,\ldots,w_n) \in W(n,M,\tau)$ and the distance distribution of $C$ with respect
to $c^\prime =(0,0,\ldots,0) \in H(n-1,2)$ be $w^\prime=(w^\prime_0,w^\prime_1,\ldots,w^\prime_{n-1}) \in W(n-1,M,\tau)$.

For every $i \in \{0, 1, \ldots, n\}$ the matrix which consists of the rows of $C$ of weight $i$ is called $i$-block.
It follows from the above notations that the cardinality of the $i$-block is $w_i$.
Next we denote by $x_i$ ($y_i$, respectively) the number of the ones (zeros, respectively) in the
intersection of the first column of $C$ and the rows of the $i$-block.

\begin{thm}\label{th 3.1.2}
The numbers $x_i$ and $y_i$, $i=0,1,\ldots,n$, satisfy the following system of linear equations
\begin{equation}\label{sys1_int.}
\left|
\begin{array}{l}
x_i+y_i=w_i, \ i=1,2,\ldots,n-1 \\
x_{i+1}+y_i=w^{\prime}_i, \ i=0,1,\ldots,n-1 \\
y_0 = w_0 \\
x_n = w_n \\
x_i, y_i \in \mathbb{Z}, \ x_i \geq 0, \ y_i \geq 0, \ i=0,1,\ldots,n
\end{array}
\right. .
\end{equation}
\end{thm}

\begin{proof}
The equalities $x_i+y_i=w_i$, $i= 1,\ldots,n-1$, $x_n=w_n$ and $y_0=w_0$ follow directly
from the definition of the numbers $x_i$ and $y_i$. The relations $x_{i+1}+y_i=w^{\prime}_i$,
$i=0,1,\ldots,n-1$, connecting $w$ and $w^\prime$, follow from the fact that the
rows of $C^\prime$, which are at distance $i$ from $c^{\prime}$, are obtained
in two ways: from the $y_i$ rows of $C$ at distance $i$ from $c$ and first coordinate $0$,
and from the $x_{i+1}$ rows of $C$ at distance $i+1$ from $c$ and first coordinate $1$. \hfill $\Box$
\end{proof}

\begin{cor}
\label{algo-a0}
The distance distribution $w=(w_0,w_1,\ldots,w_n) \in W(n, M, \tau)$ is ruled out
if no system (\ref{sys1_int.}) obtained when $w^\prime$ runs $W(n-1,M,\tau)$ has a solution.
\end{cor}

\begin{rem}
\label{BK-theoremA}
Theorem \ref{th 3.1.2} was firstly proved and used in 2013 by Boyvalenkov-Kulina \cite{BK1} for $w \in P(n,M,\tau)$.
\end{rem}

Corollary \ref{algo-a0} rules out some distance distributions $w$ but it mainly serves to produce
feasible pairs $(w,w^\prime)$ which will be investigated further.

Our next step is based on the following property of BOAs: if we take the rows of $C$ with first coordinate
$0$ ($1$, respectively) and remove that first coordinate then we obtain a BOA $C_0$ ($C_1$, respectively)
of parameters $(n-1,M/2,\tau-1)$ (see Figure 1). At this stage the BOAs $C_0$ and $C_1$ have the same sets of
admissible distance distributions -- all these which have passed the sieves of Corollaries \ref{complementary}
and \ref{algo-a0} for the set $W(n-1,M/2,\tau-1)$.

\begin{center}
{\bf Figure 1.}
\label{construction A}
\smallskip

\begin{tabular}{c}
  $w^\prime=(w^{\prime}_0,w^{\prime}_1,\ldots,w^{\prime}_{n-1})$ \\
  $C^{\prime} = (n-1,M,\tau)$ \\
\begin{tabular}{|c|c|}
\hline
0 & \\
0 & $y=(y_0,y_1,\ldots,y_{n-1})$ \\
$\vdots$ & $C_0=(n-1,M/2,\tau-1)$ \\
0 & \\
\hline
1 & \\
1 & $x=(x_1,x_2,\ldots,x_n)$ \\
$\vdots$ & $C_1=(n-1,M/2,\tau-1)$ \\
1 & \\
\hline
\end{tabular} \\
$w=(w_0,w_1,\ldots,w_n)$ \\
$C=(n,M,\tau)$ \\
\end{tabular}
\end{center}

We continue with relations between the BOAs $C$, $C^{\prime}$, $C_0$ and $C_1$ using the numbers
$x_i$ and $y_i$, $i=0,1,\ldots,n$.

\begin{thm}
\label{a1}
The distance distribution of the $(n-1,M/2,\tau-1)$ BOA $C_0$ with respect to $c^{\prime}$ is $y=(y_0,y_1,\ldots,y_{n-1})$,
i.e. $y \in W(n-1,M/2,\tau-1)$.
\end{thm}

\begin{proof}
The number $y_i$ is equal to the number of the points of $C_0$ at distance $i$ from the point $c^{\prime}$. \hfill $\Box$
\end{proof}

\begin{rem}
\label{BKMS-theorem1a}
We have two possibilities in Theorem \ref{a1} -- if $y_0 \geq 1$, then $c^{\prime} \in C_0$ and
therefore $y \in P(n-1,M/2,\tau-1)$ (this is Theorem 1a) in \cite{BKMS}), or $y_0 = 0$ when
$c^{\prime} \not\in C_0$ and therefore $y \in Q(n-1,M/2,\tau-1)$.
\end{rem}

\begin{cor}
\label{a2}
The pair $(w,w^\prime)$ is ruled out if $y \not\in W(n-1,M/2,\tau-1)$ or if
$\overline{y}=(y_{n-1},y_{n-2},\ldots,y_0) \not\in W(n-1,M/2,\tau-1)$.
\end{cor}

\begin{proof}
Follows from Theorem \ref{a1} and Corollary \ref{complementary} for $\overline{C_0}$. \hfill $\Box$
\end{proof}

\begin{thm}
\label{a3}
The distance distribution of the $(n-1,M/2,\tau-1)$ BOA $C_1$ with respect to $c^{\prime}$ is $x=(x_1,x_2,\ldots,x_n)$,
i.e. $x \in W(n-1,M/2,\tau-1)$.
\end{thm}

\begin{proof}
The number $x_i$ is equal to the number of the points of $C_1$ at distance $i-1$ from the point $c^{\prime}$. \hfill $\Box$
\end{proof}

\begin{rem}
\label{BKMS-theorem2a}
Similarly to above, we have two possibilities in Theorem \ref{a3} -- if $x_1 \geq 1$, then $c^{\prime} \in C_1$ and
therefore $x \in P(n-1,M/2,\tau-1)$ (this is Theorem 2a) in \cite{BKMS}), or $x_1=0$ when
$c^{\prime} \not\in C_1$ and therefore $x \in Q(n-1,M/2,\tau-1)$.
\end{rem}

\begin{cor}
\label{a4}
The pair $(w,w^\prime)$ is ruled out if $x \not\in W(n-1,M/2,\tau-1)$ or if
$\overline{x}=(x_n, x_{n-1}, \ldots, x_1) \not\in W(n-1,M/2,\tau-1)$.
\end{cor}

\begin{proof}
Follows from Theorem \ref{a3} and Corollary \ref{complementary} for $\overline{C_1}$. \hfill $\Box$
\end{proof}

In our next step we consider the effect of the permutation $(0 \to 1, 1 \to 0)$ in the first column of
$C$. This transformation does not change the distances from $C$ and thus we obtain a BOA $C^{1,0}$ of
parameters $(n,M,\tau)$ again.

\begin{thm}
\label{a5}
If the distance distribution of $C$ with respect to $c=\mathbf{0} \in H(n,2)$ is
$w=(w_0,w_1,\ldots,w_{n-1},w_n)=(y_0,x_1+y_1,\ldots,x_{n-1}+y_{n-1},x_n)$, then
the distance distribution of $C^{1,0}$ with respect to $c$ is
$\widehat{w}=(x_1,x_2+y_0,\ldots,x_n+y_{n-2},y_{n-1})$, i.e. $\widehat{w} \in W(n,M,\tau)$.
\end{thm}

\begin{proof}
There are $x_i$ points in $C^{1,0}$ (coming from $C_1$) at distance $i-1$ from $c$.
Analogously, there are $y_i$ points in $C^{1,0}$ (coming from $C_0$) at distance $i+1$ from $c$.
This means that the number of the points of $C^{1,0}$ at distance $0$ from $c$ is $x_1$,
the number of the points of $C^{1,0}$ at distance $i$, $1 \leq i \leq n-1$, from $c$ is $y_{i-1}+x_{i+1}$,
and, finally, the number of the points of $C^{1,0}$ at distance $n$ from $c$ is $y_{n-1}$.
Therefore the distance distribution of $C^{1,0}$ with respect to $c$ is
$\widehat{w}=(x_1,x_2+y_0,\ldots,x_n+y_{n-2},y_{n-1})$. \hfill $\Box$
\end{proof}

\begin{cor}
\label{a6}
The pair $(w,w^\prime)$ is ruled out if
$\widehat{w} \not\in W(n,M,\tau)$ or if $\overline{\widehat{w}} \not\in W(n,M,\tau)$.
\end{cor}


\begin{cor}
\label{a7}
The distance distribution $w$ is ruled out if all possible pairs  $(w,w^\prime)$, $w^\prime \in W(n-1,M,\tau)$,
are ruled out.
\end{cor}

Otherwise, we proceed with the remaining pairs as follows. Let
\begin{equation}
\label{solx-y}
(x_0^{(j)}=0,x_1^{(j)},\ldots,x_n^{(j)};y_0^{(j)},y_1^{(j)},\ldots,y_{n-1}^{(j)},y_n^{(j)}=0), \ \ j = 1, \ldots, s,
\end{equation}
are all solutions of Theorem {\ref{th 3.1.2}} when $w^\prime$ runs $W(n-1,M,\tau)$ which have passed the sieves of Corollaries
\ref{a2}, \ref{a4} and \ref{a6}. We now free the cutting and
thus consider all possible $n$ cuts of columns of $C$. These cuts produce pairs $(w,w^\prime)$
(where $w$ is fixed) and corresponding solutions \eqref{solx-y}. Let the solutions \eqref{solx-y}
appear with multiplicities $k_1,k_2,\ldots,k_s$, respectively.

\begin{thm}\label{th 3.1.3}
The nonnegative integers $k_1,k_2,\ldots,k_s$ satisfy the equations
\begin{equation}\label{sys2_int.}
\left|
\begin{array}{lllcll}
k_1           & + k_2           & + & \cdots & + k_s           & = n \\
k_1 x_1^{(1)} & + k_2 x_1^{(2)} & + & \cdots & + k_s x_1^{(s)} & = w_1 \\
k_1 x_2^{(1)} & + k_2 x_2^{(2)} & + & \cdots & + k_s x_2^{(s)} & = 2 w_2 \\
&&& \ddots \\
k_1 x_n^{(1)} & + k_2 x_n^{(2)} & + & \cdots & + k_s x_n^{(s)} & = n w_n \\
\end{array}
\right. .
\end{equation}
\end{thm}

\begin{proof}
This follows for counting in two ways the number of the ones in the $i$-block of $C$.
For fixed $i \in \{1,2,\ldots,n\}$, this number is obviously $iw_i$, and, on the other hand, it is equal
to the sum $k_1x_i^{(1)}+k_2x_i^{(2)}+\cdots+k_sx_i^{(s)}$. \hfill $\Box$
\end{proof}

\begin{cor}
\label{a8}
The distance distribution $w$ is ruled out if the system (\ref{sys2_int.}) does not have solutions.
\end{cor}

\begin{cor}
\label{a9}
Let $j \in \{1,2,\ldots,s\}$ be such that all solutions of the system (\ref{sys2_int.}) have $k_j=0$.
Then the pair $(w,w^\prime)$, which corresponds to $j$, is ruled out.
\end{cor}

\section{Our algorithm}
\label{sec:2}

We organize the results from the previous section to work together as follows.

All BOAs (in fact, their current sets of feasible distance distributions $P$, $Q$ and $W$) of interest
for the targeted BOA $C=(n,M,\tau)$ are collected in a table starting with first row
\[  (\tau,M,\tau) \ (\tau+1,M,\tau) \ (\tau+2,M,\tau) \ \ldots \ C=(n,M,\tau). \]
The next row consist of the derived BOAs
\[  (\tau-1,M/2,\tau-1) \ (\tau,M/2,\tau-1) \ (\tau+1,M/2,\tau-1) \ \ldots \ (n-1,M/2,\tau-1) \]
and so on until it makes sense. We apply Corollaries \ref{algo-a0}, \ref{a7} and \ref{a8} in every row separately
from left to right to reduce the sets $P$, $Q$ and $W$. Of course, this process is fueled with information from
the columns (starting from the bottom end) according to Corollaries \ref{a2}, \ref{a4}, \ref{a6} and \ref{a9}.
Every nonsymmetric distance distribution $w$ which is ruled out, forces its mirror image $\overline{w}$ to be ruled out
according to Corollary \ref{complementary}.

The algorithm stops when no new rulings out are possible. An entry at the right end, showing that some of the sets $P$, $Q$
and $W$ is empty\footnote{In fact, in all cases where we arrived at
an empty set, the other two also became empty at the same step.},
means nonexistence of the corresponding BOA.  Otherwise, we collect the reduced sets
for further analysis and classification results (in some cases, possibly, uniqueness).

Here is the pseudocode of the module of our algorithm which deals with the sets $W(n, M, \tau)$,
$W(n-1, M, \tau)$ and $W(n - 1, M/2, \tau -1)$.

\begin{algorithm}
\caption{}
\label{ndda}
\begin{algorithmic}[1]
\Procedure{NDDA}{$W(n, M, \tau)$, $W(n-1, M, \tau)$, $W(n - 1, M/2, \tau -1)$}

\State $\textit{filteredW} =$ empty set

\For { $w \in W(n, M, \tau)$}
  \State $\textit{allX} =$  empty set

  \For { $w' \in W(n-1, M, \tau)$}

	\State {$x, y =$ solve system (\ref{sys1_int.}) for integer nonnegative solutions}
	
	\If {$x,  \bar{x} \in W(n-1, M/2, \tau-1)$ and
		$y, \bar{y} \in W(n-1, M/2, \tau-1)$ and
		$\widehat{w}, \bar{\widehat{w}} \in W(n, M, \tau)$ and
		$\widehat{w}, \bar{\widehat{w}} \not \in filteredW$
	}
	  \State {add $x$ to $allX$}
	\EndIf
  \EndFor

  \If {allX is empty}
    \State {add $w$ to $filteredW$ }
  \Else
    \If {system (\ref{sys2_int.}) has no integer nonnegative solutions}
      \State {add $w$ to $filteredW$ }
    \EndIf
  \EndIf

\EndFor

\If {$filteredW$ is nonempty}
  \State \Return{\Call{NDDA}{$W(n, M, \tau) \setminus filteredW$, $W(n - 1, M,
  \tau)$, $W(n - 1, M/2, \tau - 1)$}}
\Else
  \State \Return {$W(n, M, \tau)$}
\EndIf
\EndProcedure

\end{algorithmic}
\end{algorithm}

We believe that the above description is enough for smooth reproduction of our algorithm. Anyway, we
are ready to supply the interested reader with all our programs and databases \cite{BMSres}.

\section{New nonexistence results}
\label{sec:3}
\subsection{Nonexistence of $(9,96,4)$ BOA and consequences}
\label{sec:3.1}

We apply the algorithm from the previous section on the table below
targeting the $(9,96,4)$ BOA.

\begin{center}
\begin{tabular}{cccc}
  $(4,96,4)$ & $(5,96,4)$ & $\cdots$ & $(9,96,4)$ \\
  $(3,48,3)$ & $(4,48,3)$ & $\cdots$ & $(8,48,3)$ \\
  $(2,24,2)$ & $(3,24,2)$ & $\cdots$ & $(7,24,2)$ \\
  \end{tabular}
\end{center}

The frame of the implementation is showed in the next table. In every
entry we first show the number of distance distributions in the beginning and then (after the arrow)
the number of the remaining distance distributions in the end of the implementation. The numbers in
the brackets show how many distance distributions were left possible after \cite{BKMS}.

{\footnotesize
\begin{center}
\begin{tabular}{ccccccc}
  $P(n,96,4)$ & $1 \to 1$ & $6 \to 6$ & $12 \to 12$ & $20(10) \to 10$ & $34(12) \to 9$  & $37(10) \to 0$  \\
  $Q(n,96,4)$ & $0 \to 0$ & $1 \to 1$ &  $4 \to 4$  &     $12 \to 6$  &     $41 \to 11$ &     $97 \to 0$  \\
  $W(n,96,4)$ & $1 \to 1$ & $7 \to 7$ & $16 \to 16$ &     $32 \to 16$ &     $75 \to 20$ &    $134 \to 0$  \\
\hline
  $P(n,48,3)$ & $1 \to 1$ & $6 \to 6$ & $13 \to 13$ & $31(25) \to 25$ & $53(41) \to 38$ & $96(65) \to 62$  \\
  $Q(n,48,3)$ & $0 \to 0$ & $1 \to 1$ &  $4 \to 4$  &     $13 \to 9$  &     $41 \to 30$ &    $110 \to 85$  \\
  $W(n,48,3)$ & $1 \to 1$ & $7 \to 7$ & $17 \to 17$ &     $44 \to 34$ &     $94 \to 68$ &    $206 \to 147$ \\
\hline
  $P(n,24,2)$ & $1 \to 1$ & $6 \to 6$ & $13 \to 13$ & $30 \to 28$ &  $49 \to 47$ &  $74 \to 69$  \\
  $Q(n,24,2)$ & $0 \to 0$ & $1 \to 1$ &  $5 \to  5$ & $19 \to 17$ &  $54 \to 52$ & $130 \to 125$ \\
  $W(n,24,2)$ & $1 \to 1$ & $7 \to 7$ & $18 \to 18$ & $49 \to 45$ & $103 \to 99$ & $204 \to 194$ \\
  \end{tabular}
\end{center}
}

\medskip

\begin{thm}\label{th 4.1.1}
There exist no binary orthogonal arrays of parameters $(9,96,4)$.
\end{thm}

\begin{proof}
The zero entries in the right upper cells of the last table imply that there exists
no binary orthogonal array of parameters $(9,96,4)$. \hfill $\Box$
\end{proof}

\medskip

The implementation of the algorithm for Theorem \ref{th 4.1.1} created database which
is available at \cite{BMSres}. Note that intermediate results are also included.

In 1966 Seiden and Zemach \cite{SZ} (see also \cite[Theorem 2.24]{HSS}) proved that
BOAs of parameters $(n,N,\tau=2k)$ and $(n+1,2N,\tau+1=2k+1)$ coexist. Therefore we have
the following nonexistence result as well.

\medskip

\begin{cor}\label{cor 4.1.2}
There exist no binary orthogonal arrays of parameters $(10,192,5)$.
\end{cor}

\medskip

The last Corollary follows also from the implementation of our algorithm with $(10,192,5)$ as target.
This is illustrated in the next table.

{\footnotesize
\begin{center}
\begin{tabular}{ccccccc}
  $P(n,192,5)$ & $1 \to 1$ & $6 \to 6$ & $12 \to 12$ & $21(8) \to 8$  & $35(6) \to 4$ & $35(4) \to 0$  \\
  $Q(n,192,5)$ & $0 \to 0$ & $1 \to 1$ &  $4 \to 4$  &    $12 \to 4$  &    $32 \to 4$ &    $85 \to 0$  \\
  $W(n,192,5)$ & $1 \to 1$ & $7 \to 7$ & $16 \to 16$ &    $33 \to 12$ &    $67 \to 8$ &   $120 \to 0$  \\
 \end{tabular}
\end{center}
}

The nonexistence results of Theorem \ref{th 4.1.1} and Corollary \ref{cor 4.1.2} give improvements in two entries of
Table 12.1 from \cite{HSS}. We have $7 \leq L(n,\tau) \leq 8$ instead of $6 \leq L(n,\tau) \leq 8$ for the
pairs $(n,\tau)=(9,4)$ and $(10,5)$.

\subsection{Nonexistence of $(10,112,4)$ BOA and consequences}
\label{subsec:3.2}

Here we work in the table with $C=(11,112,4)$ as target.

\begin{center}
\begin{tabular}{cccc}
  $(4,112,4)$ & $(5,112,4)$ & $\cdots$ & $(11,112,4)$  \\
  $(3,56,3)$  & $(4,56,3)$  & $\cdots$ & $(10,56,3)$   \\
  $(2,28,2)$  & $(3,28,2)$  & $\cdots$ & $(9,28,2)$    \\
\end{tabular}
\end{center}

The results are shown below. Again, in every
entry we first show the number of distance distributions in the beginning and then (after the arrow)
the number of the remaining distance distributions in the end of the implementation.
The numbers in the brackets show how many distance distributions were left possible after the implementation
of the algorithm from \cite{BKMS}.

{\footnotesize
\begin{center}
\begin{tabular}{cccccc}
 $P(n,112,4)$ & $1 \to 1$ & $7 \to 7$ & $15(13) \to 13$ & $31(20) \to 12$ & $58(25) \to 16$ \\
 $Q(n,112,4)$ & $0 \to 0$ & $1 \to 1$ &      $5 \to 3$  &     $17 \to 6$  &     $59 \to 18$ \\
 $W(n,112,4)$ & $1 \to 1$ & $8 \to 8$ &     $20 \to 16$ &     $48 \to 18$ &    $117 \to 34$ \\
\hline
  $P(n,56,3)$ & $1 \to 1$ & $7 \to 7$ & $17(16) \to 16$ & $49(40) \to 40$ & $95(71) \to 68$  \\
  $Q(n,56,3)$ & $0 \to 0$ & $1 \to 1$ &      $4 \to 3$  &     $15 \to 14$ &     $59 \to 44$  \\
  $W(n,56,3)$ & $1 \to 1$ & $8 \to 8$ &     $21 \to 19$ &     $64 \to 54$ &    $154 \to 112$ \\
\hline
  $P(n,28,2)$ & $1 \to 1$ & $7 \to 7$ & $17 \to 17$ & $46(43) \to 43$ & $87(82) \to 82$  \\
  $Q(n,28,2)$ & $0 \to 0$ & $1 \to 1$ &  $5 \to  4$ &     $23 \to 22$ &     $79 \to 76$  \\
  $W(n,28,2)$ & $1 \to 1$ & $8 \to 8$ & $22 \to 21$ &     $69 \to 65$ &    $166 \to 158$ \\
  \end{tabular}
\end{center}
}
{\footnotesize
\begin{center}
\begin{tabular}{ccc}
 $72(28) \to 9$  & $88(31) \to 0$  \\
 $158 \to 24$ &    $373 \to 0$   \\
 $230 \to 33$ &    $461 \to 0$   \\
\hline
 $199(137) \to 135$ & $311(205) \to 193$  \\
 $181 \to 129$ &      $451 \to 313$  \\
 $380 \to 264$ &      $762 \to 506$  \\
\hline
 $145(137) \to 137$ & $208(196) \to 196$  \\
 $205 \to 195$ &      $469 \to 450$  \\
 $350 \to 332$ &      $677 \to 646$  \\
  \end{tabular}
\end{center}
}

The zero entries in the upper right corner imply the nonexistence of BOAs
of parameters $(10,112,4)$.

\medskip

\begin{thm}\label{th 5.1.1}
There exist no binary orthogonal arrays of parameters $(10,112,4)$ and $(11,112,4)$.
\end{thm}

\begin{proof}
The zero entry in the right upper cell of the last table means that there exists
no binary orthogonal array of parameters $(10,112,4)$. This immediately implies the nonexistence
of BOAs of parameters $(11,224,4)$. \hfill $\Box$
\end{proof}

\medskip

The data from the implementation of the algorithm for Theorem \ref{th 5.1.1}
is available at \cite{BMSres} with intermediate results included.

As above, we use the coexistence of $(n,N,2k)$ and $(n+1,2N,2k+1)$ BOAs to obtain further nonexistence
results.

\medskip

\begin{cor}\label{cor 5.1.2}
There exist no binary orthogonal arrays of parameters $(11,224,5)$ and $(12,224,5)$.
\end{cor}

\medskip

The last Corollary follows also from the implementation of our algorithm with $(11,224,5)$ as target.
This is illustrated in the next table, where the first two columns are missed.

{\footnotesize
\begin{center}
\begin{tabular}{cccccc}
 $P(n,224,5)$ & $15(11) \to 11$ & $32(19) \to 4$ & $63(15) \to 5$ & $74(11) \to 2$ & $108(6) \to 0$  \\
 $Q(n,224,5)$ &      $4 \to 2$  &     $16 \to 2$ &     $47 \to 4$ &    $141 \to 4$ &    $337 \to 0$  \\
 $W(n,224,5)$ &     $19 \to 13$ &     $48 \to 6$ &    $110 \to 9$ &    $215 \to 6$ &    $445 \to 0$  \\
 \end{tabular}
\end{center}
}

The nonexistence results of Theorem \ref{th 5.1.1} and Corollary \ref{cor 5.1.2} give improvements in four entries of
Table 12.1 from \cite{HSS}. We have $L(n,\tau) = 8$ instead of $7 \leq L(n,\tau) \leq 8$ for the
pairs $(n,\tau)=(10,4)$, $(11,4)$, $(11,5)$ and $(12,5)$.

\subsection{Other nonexistence results}
\label{subsec:3.3}

Our algorithm gives other nonexistence results which however are superseded by
the result of Khalyavin \cite{Khal} from 2010. We list these in the next assertion.

\medskip

\begin{thm}
\label{th 6.1.1} (\cite{Khal} and our algorithm)
There exist no binary orthogonal arrays of parameters $(10,7.2^{6}=448,6)$, $(11,7.2^{7}=896,7)$,
$(12,10.2^{8}=2560,8)$, $(13,10.2^{9}=5120,9)$, $(12,11.2^{8}$ $=2816,8)$, $(13,11.2^{9}=5632,9)$ and $(15,13.2^{10}=13312,10)$.
\end{thm}

\medskip

\section{Updated table for $L(n,\tau)$}
\label{sec:4.}

We present an updated version of the situation with the possible values of function
$L(n,\tau)$ -- the minimum possible index $\lambda$ of an $(n,M=\lambda 2^\tau,\tau)$
binary orthogonal array. Our table covers the range $4 \leq n \leq 16$ and $4 \leq \tau \leq 10$ (see Table 1).

\begin{table}
\begin{center}
\caption{(see Table 12.1 in \cite{HSS})
Minimum possible index $\lambda$ of binary orthogonal array of length $n$, $4 \leq n \leq 16$,
and strength $\tau$, $4 \leq \tau \leq 10$.}
\label{tab:1}
\smallskip

\begin{tabular}{|c|c|c|c|c|c|c|c|}
      \hline\noalign{\smallskip}
      $n$ / $\tau$ & 4 & 5 & 6 & 7 & 8 & 9 & 10 \\ \hline\noalign{\smallskip}
      4 & 1 &   &   & & & & \\ \hline\noalign{\smallskip}
      5 & 1 & 1 &   & & & & \\ \hline\noalign{\smallskip}
      6 & 2 & 1 & 1 & & & & \\ \hline\noalign{\smallskip}
      7 & $^{sz}4$     & 2 & 1 & 1 &   & & \\ \hline\noalign{\smallskip}
      8 & $4^{c}$        & $^{sz}4$     & 2 & 1 & 1 & & \\ \hline\noalign{\smallskip}
      9 & $^{bms}$7-8    & $4^{c}$        & 4 & 2 & 1 & 1 & \\ \hline\noalign{\smallskip}
     10 & $^{bms}8$      & $^{bms}$7-8    & $^{kh}8$  & 4 & 2 & 1 & 1 \\ \hline\noalign{\smallskip}
     11 & $^{bms}8$      & $^{bms}8$      & $8^c$     & $^{kh}8$  & 4 & 2 & 1 \\ \hline\noalign{\smallskip}
     12 & $^{bkms}8$     & $^{bms}8$      & 12-16     & $8^c$     & $^{kh}8$  & 4 & 2 \\ \hline\noalign{\smallskip}
     13 & 8              & $^{bkms}8$     & 16        & 12-16     & $^{kh}16$ & $^{kh}8$  & 4 \\ \hline\noalign{\smallskip}
     14 & 8              & 8              & 16        & 16        & $16^c$    & $^{kh}16$ & $^{kh}8$ \\ \hline\noalign{\smallskip}
     15 & $8^{\tiny nr}$ & 8              & $16^{rh}$ & 16        & 26-32     & $16^c$    & $^{kh}16$ \\ \hline\noalign{\smallskip}
     16 & 10-16          & $8^{\tiny nr}$ & 21-32     & $16^{rh}$ & 39-64     & 26-32     & $^{kh}32$ \\
     \hline\noalign{\smallskip}
\end{tabular}

\smallskip

{\bf Key:}

\begin{tabular}{ll}
$bkms$ & Boyvalenkov, Kulina, Marinova, Stoyanova in \cite{BKMS} \\
$c$  & Cyclic code \\
$kh$ & Khalyavin (2010) (and see also Section 6 and \cite[Section 4.3]{BKMS} and ) \\
$nr$ & Nordstrom-Robinson (1967) code \\
$rh$ & Rao-Hamming construction \\
$sz$ & Seiden and Zemach (1966) bound \\
$bms$ & Theorem \ref{th 4.1.1} and Corollary \ref{cor 4.1.2},
or Theorem \ref{th 5.1.1} and Corollary \ref{cor 5.1.2}
\end{tabular}
\end{center}
\end{table}

All calculations in this paper were performed by programs in Maple.
All results (in particular all possible distance distributions in the beginning)
can be seen at \cite{BMSres}. All programs are available upon request.




\end{document}